\documentclass[runningheads]{llncs}

\usepackage{amsmath}
\usepackage{tikz}
\usepackage{xspace}
\newcommand{\junk}[1]{}

\newcommand{\sfsize}{\fontsize{0.74\baselineskip}{0.74\baselineskip}\selectfont}
\newcommand{\sans}[1]{\textsf{\sfsize \mbox{#1}}}
\newcommand{\codename}{\sans{HaPPY-Mine}\xspace}

\newenvironment{LabeledProof}[1]{\noindent{\bf Proof of #1: }}{\qed}      

\author{
}
\institute{}

\title{\codename:  Designing a Mining Reward Function}
\author{Lucianna Kiffer \and Rajmohan Rajaraman 
}

\institute{Northeastern University, Boston MA,USA \\
	\email{\{lkiffer,rraj\}@ccs.neu.edu }
	\footnote{We thank the anonymous reviewers and Yonatan Sompolinsky for their helpful comments. The first author was supported by a Facebook Fellowship and Dfinity Scholarship. This work was also partially supported by NSF grant CCF-1909363. This work was initiated when the first author was at an internship at DAGlabs. This paper appeared in Financial Cryptography and Data Security 2021.} 
}

\begin{document}
\maketitle
\vspace{-0.3in}
\begin{abstract} 

In cryptocurrencies, the block reward is meant to serve as the incentive mechanism for miners to commit resources to create blocks and in effect secure the system. 
Existing systems primarily divide the reward in proportion to expended resources and follow one of two static models for total block reward: (i) a fixed reward for each block (e.g., Ethereum), or (ii) one where the block reward halves every set number of blocks (e.g., the Bitcoin model of halving roughly every 4 years) but otherwise remains fixed between halvings.
In recent work, a game-theoretic analysis of the static model under asymmetric miner costs showed that an equilibrium always exists and is unique \cite{arnosti2018bitcoin}.  Their analysis also reveals how asymmetric costs can lead to large-scale centralization in blockchain mining, a phenomenon that has been observed in Bitcoin and Ethereum and highlighted by other studies including \cite{gervais2014bitcoin,leonardos2020oceanic}. 

In this work we introduce a novel family of mining reward functions, \codename (HAsh-Pegged Proportional Yield), which peg the value of the reward to the hashrate of the system, decreasing the reward as the hashrate increases. 
\codename distributes rewards in proportion to expended hashrate and inherits the safety properties of the generalized proportional reward function established in \cite{chen2019axiomatic}.
We study \codename under a heterogeneous miner cost model and show that an equilibrium always exists with a unique set of miner participants and a unique total hashrate.
Significantly, we prove that a \codename equilibrium is more decentralized than the static model equilibrium under a set of metrics including number of mining participants and hashrate distribution. 
Finally, we show that any \codename \textit{equilibrium} is also safe against  collusion and sybil attacks, and explore how the market value of the currency affects the equilibrium. 

\end{abstract}

 \section{Introduction}
Existing cryptocurrencies rely on block rewards for two reasons: to subsidize the cost miners incur securing the blockchain and to mint new coins.
Miners in major cryptocurrencies like Bitcoin and Ethereum participate in the protocol by packaging user transactions into blocks and incorporating those blocks into the blockchain (the global record of all transactions that have taken place in the system). 
Creating a block involves significant computational power where the miner preforms iterations of some kind of computation, the \textit{proof of work}, generally iterating over a hash function. 
This \textit{work}, whether on a CPU, GPU or other specialized hardware, comes at a cost to the miner. 
To compensate miners for incurring this cost and to incentivize more miners to join, miners collect a \textit{block reward} of newly minted coins for each block that gets added to the blockchain.
In expectation, miners are rewarded in proportion to the resources they contribute. 
This computational work is also what cryptographically ties each block in the blockchain together and makes it so that anyone wanting to \textit{fork} the blockchain, i.e. erase transactions by creating their own version of a subset of the chain, would have to redo an equivalent amount of work. The more resources miners invest in the system, the greater the system hashrate, the more expensive this attack becomes. In effect, the computational work of miners secures the blockchain system by making the blockchain immutable. 

There are two common frameworks for the block reward function in terms of distribution of supply. Bitcoin's protocol has a set maximum number of coins that will ever be minted, therefore the mining reward diminishes over time. The mining reward halves every 210,000 blocks (approximately every 4 years). For now, miners continue to profit since the value of each Bitcoin has increased over time making up for the decrease in reward with increases system hashrate. Eventually though, the mining reward will reach zero and miners will be repaid solely in transaction fees for the transactions they include in the blocks they mine. Another cryptocurrency, Ethereum, currently has in its protocol a fixed mining reward of 5 Ethers for all blocks ever. This means that the supply of Ether is uncapped and the mining hashrate can grow linearly in the market value of Ether.

In general miner costs are asymmetric \cite{miningcost} with miners with access to low-cost electricity or mining hardware being at an advantage. This has led to large centralization in both Bitcoin and Ethereum mining, with a significant portion of the hashrate being controlled by a few mining pools \cite{bitcoinpool,gervais2014bitcoin,ethereumpool}. This prevents other players from having a share of the market. We ask the question, can we design a mining reward function that alleviates these problems?

\subsection{Main contributions}
In this paper, we develop a novel hashrate-based mining reward function,{\sans{HaPPY-}\xspace} {\sans{Mine}\xspace}, which sets the block reward based on the system hashrate. \codename is defined so that as the system hashrate increases, the block reward smoothly decreases. We now outline the main contributions of this paper.
\begin{enumerate}
    \item  We introduce the notion of a \emph{hashrate-pegged mining reward function}, and formally argue that it can help in decentralizing the blockchain by reducing the hashrate that a new miner is incentivized to buy.
    \item We present \codename, a family of hashrate-pegged mining reward functions that dispense rewards in proportion to the expended hashrate.  We conduct a rigorous equilibrium analysis of the \codename family under general miner costs.  We establish that equilibria always exist, and are more decentralized than an equlibrium under the static reward function: in particular, \codename equilibria have at least as many participating miners as and lower total hashrate than an equilibrium for the static reward function.
    \item
    We show that \codename equilibria (as well as that of a static reward function) are resistant to any collusion attack involving fewer than half the miners, and that a Sybil attack does not increase the utility of the attacker.
    \item
    We finally consider the scenario where rewards are issued in the currency of the blockchain and study the effect of the change in the currency's value on the equilibrium.  We show that in \codename, an increase in the value of the cryptocurrency allows more higher cost miners to participate, again resulting in greater decentralization as compared to an equilibrium under the static reward function.
\end{enumerate}

\junk{
We 1) introduce a hash-pegged mining reward function 2) analyze the equilibrium of this function and the good properties it has compared to the static reward function 3) show it's collusion/Sybil resistance in equilibrium (ie. safety) 4) analyze the impact of a change in the currency's value to the equilibrium. Include the main takeaways from each analysis. 
}

\noindent \textbf{Outline of the paper.}  We begin in Section~\ref{background} with a description of the equilibrium analysis of \cite{arnosti2018bitcoin}, which provides a basic game-theoretic framework that we build on.  We also describe the properties satisfied by the \textit{generalized proportional allocation rule} of \cite{chen2019axiomatic}, of which our function is a special case.  In Section~\ref{hashpegged} we introduce our hash-pegged mining reward function and in Section~\ref{equilibrium_analysis} we analyze its equilibria. We analyze other factors that impact the equilibria in Section~\ref{equilibrium_impacts}.  We conclude with a discussion on the practicality of implementing the hash-pegged mining reward function in a system and with future and related work in Sections~\ref{discussion} and \ref{related_work}. 
 \section{Background}
\label{background}
In this paper, we follow a miner model of asymmetric costs with rewards being awarded in proportion to expended resources(hashrate).  Our study builds on an analysis framework developed in~\cite{arnosti2018bitcoin}.  In this section, we first summarize the model of~\cite{arnosti2018bitcoin} and their equilibrium analysis of a static reward function for mining.  We next review proportional allocation, used in both the static reward function and \codename, 
and state salient properties established in~\cite{chen2019axiomatic}.
\smallskip

\noindent {\bf Equilibrium analysis of static reward function.} 
 The \emph{simple proportional model} introduced in~\cite{arnosti2018bitcoin} has $n$ miners with costs $c_1, c_2,\dots, c_n$ where $c_1 \leq c_2 \leq \dots \leq c_n\leq \infty$. A miner $i$ who invests $q_i$ hashrate at a cost of $c_iq_i$ has mining reward and utility given by
 \[
x_i(q)= \frac{q_i}{\sum_j q_j} \mbox{ and }
 U_i(q) = x_i(q)- c_iq_i,
 \]
 respectively.  The main result of ~\cite{arnosti2018bitcoin} is that there is a \textit{unique pure strategy equilibrium} where each miner invests 
 $$q_i = \frac{1}{c^*}\max(1-c_i/c^*,0)$$
 for the unique value $c^*$ s.t. $X(c^*) = 1$ where 
 $$X(c) = \sum_i \max(1-c_i/c, 0).$$
 The value $c^*$ thus serves as a bound for which miners participate, with a miner $i$ participating if $c_i<c^*$.
  They also show that the number of miners must be finite for there to be an equilibrium strategy and that even countably infinite miners would not have an equilibrium strategy.

\smallskip

\noindent {\bf Properties of proportional allocation.}
In \cite{chen2019axiomatic}, the authors define a set of properties that allocation rules can satisfy: non-negativity, budget-balance (\textit{strong-} means all the reward is allocated, \textit{weak-} means less or all of the reward is allocated), symmetry (two miners with equal hashrate get equal reward), sybil-proofness (can't split hashrate and get more reward) and collusion-proofness (can't join hashrates and get more). They prove that the proportional allocation rule is the only rule that satisfies all of the above properties. They also define a \textit{generalized proportional allocation rule} as 
$$x_i(q)= f(\sum_j q_j)\frac{q_i}{\sum_j q_j}$$
for some function $f$ which takes in the sum of hashrate and returns the amount of reward that will be allocated. The static reward function is an example of the generalized proportional allocation rule with $f(\sum_j q_j) = 1$. In \codename, we provide a family of functions for $f$. These functions follow the generalized proportional allocation rule and, hence, satisfy all of the above properties with a \textit{weak} budget-balance as, by definition, the full reward value is not always rewarded (i.e. $f(\sum_j q_j)\leq 1$).

 \section{Hashrate-Pegged Block Reward}
\label{hashpegged}
We now introduce the notion of a \emph{hash-pegged mining reward function}. We consider a miner's decision of how much hashrate to purchase when they are joining the system.  
In this section, we consider a simplified model where the network currently has hashrate $1$ with network operational cost $c$ and mining reward of $1$ per block such that mining is profitable, i.e. $c<1$ and the \textit{system's} utility is $U = 1-c$.  
Given the network hashrate $H= \sum_j q_j$, we consider block reward 
$$r(H)=\left(\frac{1}{H}\right)^\delta$$
for a given parameter $\delta \ge 0$ such that any additional hashrate added to the system decreases the block reward\footnote{Note that our $r(H)$ function is replacing \cite{chen2019axiomatic}'s $c$ function. We change notation so as not to confuse the reward with the cost of hashrate}. 

The focus of this section is on answering the following question: Given a new miner with cost $c_i$, how much hashrate is this new miner incentivized to buy?  That is, what $q_i$ maximizes their utility
$$ U_i(q) = \frac{q_i}{1+q_i}r(1+q_i)-c_iq_i?$$

\smallskip
\noindent{\textbf{Case: $\delta = 0$, static reward.}}
First consider the fixed reward system where the reward is always 1. A new miner joining the system with hashrate $q_i$ will have utility $U_i(q)= \frac{q_i}{q_i+1}-c_iq_i$ which they want to maximize.  By solving for $U_i'(q)=0$ with $q_i>0$ and $c_i<1$, we find  that the miner maximizes their utility by buying hashrate $q_i = \sqrt{\frac{1}{c_i}}-1$.

\smallskip
\noindent{\textbf{Case: $\delta = 1$, linear decrease in reward.}}
With $r(H)= \frac{1}{H}$, a miner now wants to maximize $U_i(q)= \frac{q_i}{(q_i+1)^2}-c_iq_i$.   We can't easily solve for $U_i'(q) = \frac{1}{(q_i+1)^2}-\frac{2q_i}{(q_i+1)^3}-c_i=0$. What we can observe is that $U_i"(q)=\frac{6q_i}{(q_i+1)^4}-\frac{4}{(q_i+1)^3}$ and that $U_i"(q)<0$ for $q_i<2$, i.e. $U_i(q)$ is concave down when a miner buys less than double the current hashrate of the system.  Since $U'_i(q_i=\sqrt{\frac{1}{c_i}}-1) = 2c_i(\sqrt{c_i}-1) < 0$ for $c_i<1$, we obtain that 
for a miner that's acquiring less than twice the current system hashrate, \emph{the hashrate bought by the miner under a linearly diminishing reward ($\delta = 1$) is less than that bought under a static reward ($\delta = 0$)}. 
(For a miner buying more than twice the hashrate ($q_i\geq 2$), $c_i$ would have to be sufficiently small for this to be profitable i.e. $c_i<\frac{1}{(1+q_i)^2}<\frac{1}{9}$.)

\junk{We want to know if a new miner in this system would buy more or less hashrate than in the fixed reward setup, i.e. is this $U_i(q)$ is maximized at $q_i< \sqrt{\frac{1}{c_i}}-1$?}

\junk{
Thus to check if the maxima for this $U_i(q)$ is at $q_i<\sqrt{\frac{1}{c}}-1$ we just need to check that $U_i'(q)<0$ at $q_i = \sqrt{\frac{1}{c_i}}-1$. We get $U_i(q_i=\sqrt{\frac{1}{c_i}}-1) = 2c_i(\sqrt{c_i}-1)$ which is $<0$ for $c_i<1$ which is our initial assumption. Thus for a miner that's acquiring less than twice the current system hashrate, the miner in this linearly diminishing reward setup would buy less hashrate. 
}

\smallskip
\noindent{\textbf{General $\delta$.}}
We now analyze the impact of a more drastic decay function (larger $\delta$) on the optimal hashrate bought by a new miner joining the system.  When a new miner joins with additional hashrate $q_i$, the mining reward becomes $(\frac{1}{q_i+1})^\delta$, where $0\leq\delta< \infty$. The utility function is now $U_i(q) = \frac{q_i}{q_i+1}(\frac{1}{q_i+1})^\delta-c_iq_i = \frac{q_i}{(q_i+1)^{\delta+1}}-c_iq_i$. 
\begin{proposition}
\label{prop:diminishing_reward}
The optimal hashrate for a new miner decreases with increasing $\delta$.
\end{proposition}
Our proof proceeds in two steps.  We show that (1) the utility is a concave function at the maxima and (2) the derivative of the utility w.r.t. $q_i$ is decreasing in $\delta$. We then obtain that the utility maximum (i.e. the $q_i$ s.t. $U_i'(q) = 0$) is decreasing with an increase in $\delta$.
Due to space constraints, we defer the proof to Appendix~\ref{app:diminishing_reward}.

Thus, if we increase the $\delta$ exponent in the total block reward, we decrease the hashrate that a new miner is incentivized to buy.
While this may not have an effect for smaller miners who do not have the resources to purchase their maximal utility hashrate, Proposition~\ref{prop:diminishing_reward} demonstrates that a hash-pegged reward function can be a useful decentralization tool that disincentivizes rational big miners from joining the system with a large fraction of the hashrate. 

Note that Proposition~\ref{prop:diminishing_reward} does not take into account the dynamic game between different miner's choices. We now formally define the above family of hash-pegged mining reward functions for arbitrary system hashrate as \codename and analyze the equilibria given a set of miners with asymmetric costs. 
 \section{\codename Equilibrium Analysis}
\label{equilibrium_analysis}
Building on the model of \cite{arnosti2018bitcoin} we define a non-cooperative game between $m$ miners with cost $c_1\leq c_2 \leq \dots \leq c_m$ where each miner $i$ with hashrate $q_i$ has utility
$$ U_i(q) = x_i(q)-c_iq_i.$$
In \codename we set the maximal block reward to be $1$ and have the reward start to decrease after the system's hashrate surpasses $Q$, for a parameter $Q > 0$. We define the reward for miner $i$ as
$$ x_i(q)=\frac{q_i}{\sum_j q_j}r(q) ~~\text{where}~~ r(q) =\min\left(1, \left(\frac{Q}{\sum_j q_j}\right)^\delta\right)$$
for system parameter $\delta\in [0,\infty)$.

The main results of this section concern the existence and properties of pure Nash equilibria for the above \codename game.
\junk{We first re-write $r(q)$ as 
\begin{equation*}
r(q)= \begin{cases}
1 &\text{if $\sum_j q_j \leq Q$}\\
(\frac{Q}{\sum_j q_j})^\delta &~~~~~~~\text{o/w}
\end{cases}
\end{equation*}}
We begin our analysis by differentiating $r(q)$ and $x_i(q)$ with respect to $q_i$, and finding the derivative of $U_i(q)$ w.r.t. $q_i$.
\junk{
\begin{equation*}
r'(q)= \begin{cases}
0 &\text{if $\sum_j q_j < Q$}\\
\frac{-\delta Q^\delta}{(\sum_j q_j)^{\delta+1}} &~~~~~~~\text{if $\sum_j q_j > Q$}
\end{cases}
\end{equation*}
and deriving $x_i'(q)$.
\begin{equation*}
x_i'(q)= \begin{cases}
\frac{\sum_jq_j-q_i}{(\sum_j q_j)^2} &\text{if $\sum_j q_j < Q$}\\
\frac{Q^\delta}{(\sum_j q_j)^{\delta+2}}[\sum_jq_j-(\delta+1)q_i] &~~~~~~~\text{if $\sum_j q_j > Q$}
\end{cases}
\end{equation*}
Finally, we find the derivative of $U_i(q)$ w.r.t. $q_i$
}
\begin{equation*}
U_i'(q)= \begin{cases}
\frac{\sum_jq_j-q_i}{(\sum_j q_j)^2}-c_i &\text{if $\sum_j q_j < Q$}\\
\frac{Q^\delta}{(\sum_j q_j)^{\delta+2}}[\sum_jq_j-(\delta+1)q_i]-c_i ~~~~~~&\text{if $\sum_j q_j > Q$}
\end{cases}
\end{equation*}
Recall that for equilibria we need that $U_i'(q) \leq 0$ with equality for $q_i > 0$.  (For the case $\sum_j q_j = Q$, we need the left and right derivatives to be nonnegative and nonpositive, respectively.)

\subsection{Examples with diverse cost scenarios}
We work through some cost examples to gain intuition for the equilibrium analysis of the above reward function.

\smallskip
\noindent\textbf{Example 1}
First we consider a general 2-miner case with $\delta$ and $Q$ set to 1. In this model we have 2 miners with costs $c_1,c_2$ s.t. $c_1\leq c_2$. See Appendix~\ref{app:ex2} for the full analysis. If $c_1+c_2>1$ we use the analysis of \cite{arnosti2018bitcoin} with reward 1 and obtain that the equilibrium hashrate is $q_1+q_2<Q=1$ with $q_i = \frac{1}{c_1+c_2}(1-\frac{c_i}{c_1+c_2})$.  If $c_1+c_2 \leq 1$, then there are multiple equilibria where $\alpha+\beta =1$ with $\frac{1-c_1}{2}\leq \alpha \leq 1-c_1$ and $\frac{1-c_2}{2}\leq \beta \leq 1-c_2$. Note the equilibria system hashrate with two miners is always $\leq Q=1$.

Taking $c_1 +c_2 \leq 1$, let us consider the total utility of an equilibrium.
$$\max_{\alpha,\beta}(U_1+U_2) = \max_{\alpha,\beta}(1-c_1\alpha-c_2\beta)= \max_{\alpha}(1-c_2+(c_2-c_1)\alpha)$$
Thus, a utilitarian equilibrium is one where $\alpha$ is maximized, i.e. $\alpha = 1-c_1$. The utilitarian equilibrium is thus the one with maximal utility for the miner with least cost and lowest utility for the miner with most cost. 

\smallskip
\noindent\textbf{Example 2: $c_i = \frac{i}{i+1}$}
We now consider an example from \cite{arnosti2018bitcoin} where the cost function $c_i = \frac{i}{i+1}$, still considering $\delta = Q = 1$. This case is interesting because in the static reward case (i.e. $U_i(q)=\frac{q_i}{\sum_i q_i}-q_i c_i$) the equilibrium strategy has that $\sum_i q_i>1$ and that only the first 7 miners participate. This equilibrium point would have less reward in \codename and thus may no longer be the equilibrium point. We solve this in Appendix~\ref{app:ex4} and find that 
$$ q_i = \frac{1}{2}\sqrt{\frac{n-2}{\sum_{j=1}^n \frac{j}{j+1}}}(1-\frac{(n-2)i}{\sum_{j=1}^n \frac{j}{j+1}(i+1)})  $$
for all miners that participate in equilibrium. We can iterate over $n$ to find that with this strategy, equilibrium exists at $n=25$, i.e. for $n>25$ only the first $25$ miners participate otherwise all miners participate. Thus \codename with $\delta=1 $ results in an equilibrium with \textit{more miners participating} than in the equilibrium under a static reward function. 

\smallskip
\noindent\textbf{Example 3: $c_i = c$ for all $i$}
The next example we consider is the case of homogeneous cost with $m$ miners, $Q=1$ and any $\delta$. See Appendix~\ref{app:ex3} for the full analysis.  For $c> \frac{m-1}{m}$, we can use the analysis of \cite{arnosti2018bitcoin} and obtain $q_i = \frac{m-1}{m^2c}$ with $\sum_i q_i = \frac{m-1}{mc} < 1$. For $\frac{m-\delta -1}{m}\leq c \leq \frac{m-1}{m}$, an equilibrium exists at $\sum_i q_i = 1$ where $q_i = \frac{1}{m}$. Finally for $c<\frac{m-\delta -1}{m}$ we get an equilibrium strategy with $\sum_i q_i >1$ where $q_i = \frac{1}{m}\sqrt[\delta +1]{\frac{m-\delta -1}{cm}}$. In each case the equilibrium hashrate for \codename for any $\delta$ is less than or equal to that of the static reward equilibria.  In Corollary~\ref{corr:hashrate} below, we show this in fact holds for any set of costs. 

\subsection{General analysis of \codename}
We now analyze the equilibria for the general case of \codename with $m>\delta+1$ miners with costs $c_1\leq c_2\leq ...\leq c_m< c_{m+1}=\infty$. Recall the utility function
\begin{equation*}
	U_i(q)= \begin{cases}
		\frac{q_i}{\sum_j q_j}-q_ic_i &\text{if $\sum_j q_j \leq Q$}\\
		\frac{q_i}{\sum_j q_j}(\frac{Q}{\sum_j q_j})^\delta-q_ic_i &~~~~~~~\text{o/w}
	\end{cases}
\end{equation*}

In the propositions below we first derive necessary conditions for an equilibrium to exist in different cases depending on how the system hashrate $\sum_i q_i$ compares with $Q$. Taking these propositions we derive lemmas proving the existence of equilibria given any set of miner costs. The lemmas also prove the impossibility of equilibria to exist simultaneously for different values of $\sum_i q_i$, i.e. the uniqueness of the equilibria. We finish this section with our final theorem statement defining the equilibria values given a set of costs, as well as corollaries on the properties of the equilibria. 

\begin{proposition}[Necessary condition for equilibrium with total hashrate less than $Q$, \cite{arnosti2018bitcoin}]
	\label{prop:less}
	If $\sum_i q_i<Q$ at equilibrium then there exists a $c^*>1/Q$ such that $X(c^*)=1$ and all miners $i$ with $c_i < c^*$ participate with $q_i=\frac{1}{c^*}(1-c_i/c^*)$.
\end{proposition}
\begin{proof}
	If $\sum_i q_i<Q$ then miners have utility function $U_i(q) = \frac{q_i}{\sum_j q_j}-q_ic_i$ which is the same as the simple proportional model of \cite{arnosti2018bitcoin} where there is an equilibrium strategy with $q_i = \frac{1}{c^*}\max(1-c_i/c^*,0)$ for $c^*$ such that $X(c^*)=1$. In this analysis $\sum_j q_j = \frac{1}{c^*}$, and so for $\sum_j q_j <Q$ we have $c^*>1/Q$. 	\qed  
\end{proof}

\begin{proposition}[Necessary condition for equilibrium with total hashrate equal to $Q$]
	\label{prop:equal}
	If $\sum_i q_i=Q$ at equilibrium then all miners with cost $c_i < 1/Q$ participate and satisfy $$\frac{1}{\delta+1}(Q-c_iQ^2)\leq q_i \leq Q-c_iQ^2$$
\end{proposition}
\begin{proof}
	Assume there is an equilibrium strategy such that $\sum_i q_i=Q$. The utility of a miner $i$ is given by
	$$U_i(q) = q_i(\frac{1}{Q}-c_i) \leq 0$$
	so miners with cost $c_i>1/Q$ will not participate; those with $c_i < 1/Q$ will. 
	
	We take the $n$ miners for which $c_i \leq 1/Q$. $\sum_i q_i = Q$ is an equilibrium \textit{iff}, 
	\begin{equation*}
	U_i'(q)= \begin{cases}
	\frac{1}{Q^2}[Q-q_i]-c_i\geq 0 & \text{ for } \sum_j q_j < Q\\
	\frac{Q^\delta}{Q^{\delta+2}}[Q-(\delta+1)q_i]-c_i\leq 0 & \text{ for } \sum_j q_j > Q
	\end{cases}
	\end{equation*}
	and thus, any equilibrium strategy satisfies
	$$\frac{1}{\delta+1}(Q-c_iQ^2)\leq q_i \leq Q-c_iQ^2 $$
	Note that $c_i = 1/Q$ implies $q_i = 0$, so a miner with cost $1/Q$ does not participate.  Thus, exactly those miners with $c_i<1/Q$ participate in an equilibrium. \qed 
\end{proof}

\begin{proposition}[Necessary condition for equilibrium with total hashrate more than $Q$]
	\label{prop:more}
	If $\sum_i q_i>Q$ at equilibrium then there exists a $c^{\dagger}<1/Q$ such that $X(c^{\dagger})=\delta+1$ and all miners with cost $c_i <c^{\dagger}$ participate with 
	$$q_i = \frac{\sqrt[\delta+1]{Q^\delta}}{(\delta+1)\sqrt[\delta+1]{c^{\dagger}}}(1-c_i/c^{\dagger})$$
\end{proposition}
\begin{proof}
	Assume first there exists an equilibrium where miner $i+1$ participates and miner $i$ does not with sum of hashrate $H$. This means
	$$ U_{i+1}'(q) = \frac{Q^\delta}{H^{\delta+2}}[H-(\delta+1)q_{i+1}]-c_{i+1} = 0,
	$$
	and thus $c_{i+1} = \frac{Q^\delta}{H^{\delta+2}}[H-(\delta+1)q_{i+1}]$. For $q_i =0$ we get
	$ U_i'(q) = \frac{Q^\delta}{H^{\delta+1}}-c_{i} \leq 0$
	which means $\frac{Q^\delta}{H^{\delta+1}}\leq c_i$, putting both together we get 
	$$\frac{Q^\delta}{H^{\delta+1}}\leq c_i\leq c_{i+1} =\frac{Q^\delta}{H^{\delta+2}}[H-(\delta+1)q_{i+1}],$$
	which implies $q_{i+1}\leq 0$, a contradiction to miner $i+1$ participating. Thus in any equilibrium, if miner $i+1$ participates, then miner $i$ must also participate.
	
	Letting $H=\sum_i q_i >Q$, for a miner $i$ that participates in equilibrium 
	$$ U_i'(q) = \frac{Q^\delta}{H^{\delta+2}}[H-(\delta+1)q_{i}]-c_{i} = 0
	\implies q_i = \frac{H}{\delta+1}(1-\frac{H^{\delta+1}}{Q^\delta}c_i).$$ 
	Assuming that only the first $n$ miners participate in equilibrium, we solve for $H$
	$$ H = \sum_{i=1}^n q_i = \sum_{i=1}^n \frac{H}{\delta+1}(1-\frac{H^{\delta+1}}{Q^\delta}c_i) = \sqrt[\delta+1]{\frac{Q^\delta(n-\delta-1)}{\sum_{i=1}^n c_i}}.$$
	This also means player $n+1$ must have  $U_{n+1}'(q) \leq 0$ at $q_{n+1} =0$, so we get
	$$ U_{n+1}'(q) = \frac{Q^\delta}{H^{\delta+1}}[H-(\delta+1)q_{n+1}]-c_{n+1} = \frac{Q^\delta}{H^{\delta+1}}-c_{n+1} \leq 0,$$
	$$\implies \frac{Q^\delta}{H^{\delta+1}} = \frac{\sum_{i=1}^n c_i}{n-\delta-1} \leq c_{n+1}.$$
	Let $c^{\dagger}$ be the bound for which miners participate, i.e. miner $i$ participates iff $c_i < c^{\dagger}$. Then from the above we get that $c^{\dagger} = \frac{\sum_{i=1}^n c_i}{n-\delta-1}$.  Rewriting this and using the fact that $c_i/c^* \geq 1$ for $c_i\geq c^{\dagger}$, we obtain
	$$ \sum_i \max(1-c_i/c^{\dagger},0) = \delta+1,$$
	co-opting the $X(c)$ equation for $c^{\dagger}$ s.t. $X(c^{\dagger})=\delta+1$. Since $c^{\dagger} = \frac{Q^\delta}{H^{\delta+1}}$ it must be that $c^{\dagger}<1/Q$. Lastly we plug $c^{\dagger}$ into the equation for $q_i$ and get
	\[q_i = \frac{\sqrt[\delta+1]{Q^\delta}}{(\delta+1)\sqrt[\delta+1]{c^{\dagger}}}(1-c_i/c^{\dagger}). 	\tag*{\qed}
	\]
\end{proof}

We now use Propositions~\ref{prop:less},~\ref{prop:equal}, and~\ref{prop:more} to establish the following lemmas, which will help prove our main theorem. 
We first define $c^*$ as the value for which $X(c^*)=1$ and, for $m>\delta+1$, $c^{\dagger}$ as the value for which $X(c^{\dagger})=\delta+1$. Note that $X(c)$ is a continuous increasing function in $c$ and thus $c^*<c^{\dagger}$.

\begin{lemma}[Equilibrium when $c^*>1/Q$]
	\label{lemma:less}
	If $c^*>1/Q$, then there exists a unique equilibrium strategy with $\sum_i q_i<Q$
\end{lemma}
\begin{proof}
	We know from Proposition~\ref{prop:less} that there is an equilibrium strategy with $\sum_i q_i = \frac{1}{c^*}<Q$. Since $c^*>1/Q$ that implies $c^{\dagger}>1/Q$ so by Proposition~\ref{prop:more} there is not an equilibrium strategy with $\sum_i q_i >Q$. Finally, lets assume there is an equilibrium strategy with $\sum_i q_i = Q$. Recall from Proposition~\ref{prop:equal} that all miners with cost $<1/Q$ participate, so let $n$ be those miners s.t. $c_i < 1/Q$ for $i\leq n$. From the definition of $X(c)$ we have that $\sum_{i=1}^n 1-c_i/c^* \leq 1$ which we can solve to be $c^*(n-1)\leq \sum_{i=1}^n c_i$ and we get 
	$ \frac{n-1}{Q} <  \sum_{i=1}^n c_i$.
	From Proposition~\ref{prop:equal} we have that $q_i \leq Q-c_iQ^2$ for all $i\leq n$. Thus $\sum_{i=1}^n q_i \leq \sum_{i=1}^n Q-c_iQ^2$ which solves to
	$\sum_{i=1}^n c_i \leq \frac{n-1}{Q}$,
	and thus there is no equilibrium at $\sum_i q_i = Q$.
	\qed 
\end{proof}

\begin{lemma}[Equilibrium when $c^*\leq 1/Q \leq c^{\dagger}$]
	\label{lemma:equal}
	If $c^*\leq 1/Q \leq c^{\dagger}$ then there exists at least one equilibrium at $\sum_i q_i = Q$ and any equilibrium strategy has $\sum_i q_i = Q$ with a miner $i$ participating iff $c_i < 1/Q$.
\end{lemma}
\begin{proof}
	First, since $c^*\leq 1/Q$ we know from Proposition~\ref{prop:less} there is no equilibrium at $\sum_i q_i < Q$, and since $c^{\dagger}\geq 1/Q$ we know from Proposition~\ref{prop:more} there is no equilibrium at $\sum_i q_i > Q$. 
		Finally from Proposition~\ref{prop:equal}, for there to be an equilibrium at $\sum_i q_i = Q$ we need for each miner $i$ with $c_i < 1/Q$, $q_i$ must satisfy
	$$\frac{1}{\delta+1}(Q-c_iQ^2) \leq q_i \leq Q-c_iQ^2.$$
	Summing over all $n$ s.t. $c_i <1/Q$ for $i\leq n$, and simplifying, we derive
	\junk{
	$$ \frac{1}{\delta+1}(Qn-Q^2\sum_{i=1}^n c_i)\leq Q \leq Qn-Q^2\sum_{i=1}^n c_i$$
	must be satisfied. We can re-write the inequality that must be satisfied as }
	$$ \frac{n-\delta-1}{Q} \leq \sum_{i=1}^n c_i \leq \frac{n-1}{Q}$$
	Taking the fact that $c^* \leq 1/Q$ we get 
	$\sum_{i=1}^n 1-c_i/c^* \geq 1$
	which simplifies to $c^*(n-1)\geq \sum_{i=1}^n c_i$. 
	Taking the fact that $c^{\dagger} \geq 1/Q$ we get 
	$\sum_{i=1}^n 1-c_i/c^{\dagger} \leq \delta+1$
	which simplifies to $c^{\dagger}(n-\delta-1)\leq \sum_{i=1}^n c_i$.  Putting these together, we obtain
	$$\frac{n-1}{Q}\geq \sum_{i=1}^n c_i \geq \frac{n-\delta-1}{Q}$$\qed
\end{proof}

\begin{lemma}[Equilibrium when $c^{\dagger}<1/Q$]
	\label{lemma:more}
	If $c^{\dagger}<1/Q$ then there exists a unique equilibrium strategy with $\sum_i q_i > Q$. 
\end{lemma}
\begin{proof}
	We know from Proposition~\ref{prop:more} that there is a unique equilibrium strategy with $\sum_i q_i = \sqrt[\delta+1]{\frac{Q^\delta}{c^{\dagger}}} > Q$. Since $c^*<c^{\dagger}$ we know from Proposition~\ref{prop:less} there is not an equilibrium strategy with $\sum_i q_i <Q$. Take the $n$ miners s.t $c_i < c^{\dagger}$ for $i\leq n$. From the definition of $X(c)$ we have
	$$\sum_{i=1}^n 1-c_i/c^{\dagger} = \delta+1 \implies \sum_{i=1}^n c_i = c^{\dagger}(n-\delta-1)< \frac{n-\delta-1}{Q}.$$
	Assume there is an equilibrium with $\sum_i q_i = Q$.  By Proposition~\ref{prop:equal},  miner $i$ s.t. $c_i<1/Q$ participates with $\frac{1}{\delta+1}(Q-c_iQ^2)\leq q_i$. If there are $n$ miners s.t. $c_i<c^{\dagger}$,
	$$ \sum_{i=1}^n \frac{1}{\delta+1}(Q-c_iQ^2)\leq \sum_{i=1}^n q_i \leq Q \implies  \frac{n-\delta-1}{Q} \leq \sum_{i=1}^n c_i$$
	which is a contradiction.  Thus, there is no equilibrium with $\sum_i q_i = Q$.
	\qed 
\end{proof}

We can now put together the above lemmas to get our main result:

\begin{theorem}
	\label{thm:equillib}
	For any $\delta \in[0,\infty)$ and $m\geq2$ miners with costs $c_1\leq c_2 \leq ... \leq c_m< c_{m+1}=\infty$, let
			$$X(c) = \sum_i \max(1-c_i/c,0)$$
	and $c^*$ s.t $X(c^*)=1$ and (if $m>\delta+1$) let $c^{\dagger}$ s.t. $X(c^{\dagger}) = \delta+1$.
	\codename with $Q>0$ has equilibria as follows with system hashrate $\sum_i q_i = H$:\\
	(a) if $c^*>1/Q$, there is a unique equilibrium with $H = \frac{1}{c^*} < Q$ with
			$$ q_i = \max(\frac{1}{c^*}(1-c_i/c^*),0)$$
	(b) if $c^*\leq1/Q\leq c^{\dagger}$ or $c^*\leq1/Q$ and $m\leq \delta +1$, there exists an equilibrium and every equilibrium satisfies $H=Q$, with $q_i = 0$ for $c_i\geq 1/Q$, and otherwise
			$$ \frac{1}{\delta+1}(Q-c_iQ^2)\leq q_i \leq Q-c_iQ^2$$
	(c) if $c^{\dagger}<1/Q$, $m > \delta + 1$, there is a unique equilibrium with $H= \sqrt[\delta+1]{\frac{Q^\delta}{c^{\dagger}}}>Q$,
			$$ q_i = \max(\frac{\sqrt[\delta+1]{Q^\delta}}{(\delta+1)\sqrt[\delta+1]{c^{\dagger}}}(1-c_i/c^{\dagger}),0)$$
	 
\end{theorem}
\begin{proof}
    The case $c^*>1/Q$ follows directly from Lemma~\ref{lemma:less}. 
    Next we consider $c^*\leq 1/Q$ and $m\leq \delta+1$. Since $c^*\leq 1/Q$ we know from Proposition~\ref{prop:less} there is no equilibrium at $\sum_i q_i < Q$. For equilibria with $\sum_i q_i=H > Q$ we need that $U_i'(q) = 0$ for all miners who participate which gives us that $q_i = \frac{H}{\delta+1}[1-c_i\frac{H^{\delta+1}}{Q^\delta}]$. Assuming only the first $n$ miners participate, we get $H = \sum_i^n q_i = \sum_i^n \frac{H}{\delta+1}[1-c_i\frac{H^{\delta+1}}{Q^\delta}]$. We can simplify this to be $\frac{H^{\delta+1}}{Q^\delta}\sum_i^n c_i = n-\delta-1<0$ which is not satisfiable. The only option for equilibria is then for $\sum_i q_i = Q$ which we get from Proposition~\ref{prop:equal} iff $\frac{1}{\delta +1}[Q-Q^2c_i]\leq q_i\leq Q-Q^2c_i$ for all miners with $c_i< 1/Q$. Summing over all miners $i\leq n$ s.t $c_i<1/Q$ we get $\frac{n-\delta-1}{Q}\leq \sum_i^n c_i\leq \frac{n-1}{Q}$ must be satisfied. Notice that the left-most expression is negative so the left expression is satisfied. We know $c^*\leq 1/Q$ thus $X(1/Q) = \sum_i^n 1-c_iQ \geq 1$ which simplifies to $\sum_i^n c_i \leq \frac{n-1}{Q}$.
	Finally for $m>\delta +1$, the case for $c^*\leq 1/Q \leq c^{\dagger}$ follows from Lemma~\ref{lemma:equal} and the case for $c^{\dagger}<1/Q$ follows from Lemma~\ref{lemma:more}.
	\qed 
\end{proof}

In the following two corollaries we examine how the equilibria of \codename changes with the parameter $\delta$ in terms of miner participation and the system hashrate. In particular we show that any \codename equilibria has at least as many miners participating (with at most the same system hashrate) as in the static reward function equilibria. 

\begin{corollary}
	For any $m$ miners with costs $c_1\leq c_2 \leq ... \leq c_m$, \codename with any $Q,\delta$ has equilibria with at least as many miners participating as the static reward function.
		Furthermore, the number of miners participating in equilibria for \codename monotonically increases in $\delta$.
\end{corollary}
\begin{proof}
	By the analysis of \cite{arnosti2018bitcoin} under the simple proportional model, the static reward function has a unique equilibrium with all miners whose cost $c_i<c^*$ participating s.t $X(c^*) = 1$. \codename has at least all the same miners participating in 3 scenarios: $c_i<c^*$ for $c^*>1/Q$, $c_i<1/Q$ for $c^*\leq1/Q$ and $m\leq \delta+1$ or $1/Q\leq c^{\dagger}$ and $c_i<c^{\dagger}$ for $c^\dagger<1/Q$ where $c^*<c^{\dagger}$, i.e. in all four cases, all miners with $c_i<c^*$ are participating and possibly additional miners.
	
	For the general statement, take any $\delta$-\codename equilibrium. If $c^*>1/Q$, regardless of how you change $\delta$, $c^*$ remains fixed so by Lemma~\ref{lemma:less}, the equilibrium remain the same with the same miners. Suppose instead $c^*\leq 1/Q \leq c^\dagger$, as $\delta$ increases $c^\dagger$ increases. Thus for a larger $\delta$, the equilibrium remains at $\sum_i q_i = Q$ with the same miners of cost $c_i<1/Q$ participating.  If $c^\dagger<1/Q$, then since $c^\dagger$ acts as an upper-bound for which miners participate, as $\delta$ increases, this upper bound increases. This upper bound caps at $1/Q$; then we switch to the second equilibrium case where all miners with $c_i< 1/Q$ participate. 
	\qed 
\end{proof} 

\begin{corollary}
    \label{corr:hashrate}
   \codename has equilibria with hashrate at most that of the static reward function. Furthermore, \codename equilibria hashrate is monotonically non-increasing with an increase in $\delta$. 
\end{corollary}
\begin{proof}
    We prove the second part of the statement and note that the static reward function is \codename with $\delta = 0$, so the first statement follows. Given a set of costs, we consider the possible values of $c^*$ and $c^\dagger$. (a) If $c^*>1/Q$, then for any $\delta$, $H$ is always $1/c^*$. (b) If $c^*\leq 1/Q \leq c^\dagger$ for some $\delta$, then the equilibria hashrate for that $\delta$ is $H=Q$. As $\delta$ increases, the value of $c^\dagger$ increases so the equilibrium hashrate will continue to be $Q$ for any $\delta'>\delta$.  (c) If $c^\dagger < 1/Q$ for some $\delta$, we that $H>Q$ and we have two cases to consider for $\delta'>\delta$. Since $c^\dagger$ increases as $\delta$ increases, either it increases s.t. $c^\dagger_{new}$ becomes $\geq 1/Q$ or $m<\delta'+1$, in either case the new equilibrium hashrate would be $H'=Q<H$. The last case is that $c^\dagger<c^\dagger_{new}<1/Q$ and $m\geq \delta'+1$. In this case we first assume $H<H'$, i.e.
    \begin{align*}
    H = & \frac{Q^{\delta/(\delta+1)}}{(c^\dagger)^{1/(\delta+1)}} \\
    = & \frac{Q^{\delta'/(\delta'+1)} Q^{\delta/(\delta+1) - \delta'/(\delta'+1)}}{(c^\dagger)^{1/(\delta+1)}} \\
    \ge& \frac{Q^{\delta'/(\delta'+1)} Q^{\delta/(\delta+1) - \delta'/(\delta'+1)}}{(c_{new}^\dagger)^{1/(\delta+1)}} \tag{$c^\dagger < c^\dagger_{new}$}\\
    = & \frac{Q^{\delta'/(\delta'+1)}}{(c_{new}^\dagger)^{1/(\delta'+1)}} \frac{Q^{\delta/(\delta+1) - \delta'/(\delta'+1)}}{ (c_{new}^\dagger)^{1/(\delta+1)-1/(\delta'+1)}} \\
    = & H' \frac{Q^{(\delta - \delta')/(\delta+1)(\delta'+1)}}{ (c_{new}^\dagger)^{(\delta' - \delta)/(\delta+1)(\delta'+1)}}\\
    = & H'\left(\frac{1}{c^\dagger_{new}Q}\right)^{(\delta' - \delta)/(\delta+1)(\delta'+1)}\\
    \ge & H' \tag{$c^\dagger_{new} Q < 1$ and $\delta' > \delta$}
    \end{align*}
    \qed
\end{proof}

The previous corollaries together say that as $\delta$ increases, the number of miners who participate in equilibrium increases with the total hashrate of the system at equilibrium decreasing. We now explore what the impact of this is on the market share of miners. In particular we want to check that the new equilibrium does not disproportionately advantage lower cost miners. Unfortunately we can't make such a strong statement, owing to the presence of multiple equilibria when the sum of hashrates equals $Q$. Instead, we get the following corollary which states that for \textit{most cases}, a miner's relative market share to any higher-cost miner does not go up.  Formally, 
given two miners $i,j$ with costs $c_i<c_j$ and $\delta$ s.t. $q_i,q_j>0$ at equilibrium (i.e. both miners participate at equilibrium), we define the {\em relative market share}\/ $r_{ij}(\delta)$ as follows.  If $\sum_i q_i \neq Q$, then there is a unique equilibrium, so we define $r_{ij}(\delta)$ to be $q_i/q_j$.  Otherwise, there may be multiple equilibria and we define $r_{ij}(\delta)$ to be the ratio of the maximum value of $q_i$ to the maximum value of $q_j$ in equilibrium (defining it to be the ratio of the minimum values yields the same ratio).  
\begin{corollary}
\junk{    Given two miners $i,j$ with costs $c_i<c_j$ and $\delta$ s.t. $q_i,q_j>0$ at equilibrium (i.e. both miners participate at equilibrium). Take $\delta'>\delta$ with new equilibrium hashrates $q'_i,q'_j$.   
    If $\sum_i q'_i \neq Q$, then we have $q'_i/q'_j \geq q_i/q_j$.  If $\sum_i q'_i=Q$, then there may be multiple equilibria and the ratio of the maximum value of $q'_i$ to that of $q'_j$ and the ratio of the minimum value of $q'_i$ to that of $q'_j$ are both at least $q_i/q_j$.}
    For any two miners $i, j$ with costs $c_i < c_j$, parameters $\delta, \delta'$ such that both miners participate in equilibrium at parameter $\delta$, and $\delta' > \delta$, $r_{ij}(\delta')$ is at least $r_{ij}(\delta)$.
\end{corollary}
\begin{proof}
    Consider a miner who participates at equilibrium with a certain $\delta$. Given a set of costs, we consider the possible values of $c^*$ and $c^\dagger$. (a) If $c^*>1/Q$, then for any $\delta$, the equilibrium stays the same. (b) for $c^*\leq 1/Q \leq c^\dagger$, any increase in $\delta$ does not change this inequality and thus the equilibrium conditions do not change and thus maintain the same equilibria maximum and minimum ratios (i.e. $r_{ij}(\delta)= r_{ij}(\delta')$ for all $\delta'$) .
    
    The only interesting case is thus (c) $c^\dagger <1/Q$, as $\delta$ increases $c^\dagger$ increases. Given a $\delta'>\delta$, we compare the relative market share of two miners $i,j$ where $c_i<c_j$ as $r_{ij}(\delta')=\frac{c^{\dagger}_{new}-c_i}{c^{\dagger}_{new}-c_j}$ which is decreasing with an increase in $c^{\dagger}_{new}$ (i.e. increasing $\delta'$). Thus, while $c^{\dagger}_{new}<1/Q$, a miner's relative market share to any higher cost miner is decreasing. 
    
    The only case left to consider is a $\delta'>\delta$ s.t. $c^\dagger_{new}\geq 1/Q$. The new equilibrium hashrate $q_i'$ for miners participating is bounded by $\frac{1}{\delta+1}(Q-c_iQ^2) \leq q'_i \leq Q-c_iQ^2$. If we compare $q'_i,q'_j$ at the bounds we get $r_{ij}(\delta') = \frac{1-c_iQ}{1-c_jQ}$ which is less than the old relative market share of $\frac{1-c_i/c^\dagger}{1-c_j/c^\dagger}$ since $c^\dagger <1/Q$. \qed
    
    \end{proof}

\section{Impact of Attacks and Currency on Equilibria}
\label{equilibrium_impacts}
Our equilibrium analysis in Section~\ref{equilibrium_analysis} assumes that the number of miners and their costs are known, and that the miner costs and rewards are in the same currency unit.  In this section, we analyze certain attacks and events that may impact equilibria.  
We begin with the question: if miners are able to collude (two miners pretend to be a single miner) or duplicate themselves (a single miner pretends to be multiple miners), can they increase their own utility? In other words, are \codename equilibria resistant to miner collusion and sybil strategies?  We show that  \codename equilibria are resistant to collusion and Sybil attacks.  
We also study the effect of variable coin market value when reward is given in the coin of the blockchain.  Due to space constraints, we state the main results for collusion resistance and the effect of variable coin market value, and refer the reader to Appendix~\ref{app:equilibrium_impacts} for Sybil resistance and the missing proofs in this section. 

\smallskip
\noindent \textbf{Collusion resistance.} 
We consider the case of $m$ homogeneous miners.
\begin{lemma}
    Suppose $m$ miners with uniform costs participate in \codename with parameters $\delta,Q$.   If $k \le m/2$ of the miners collude and act as one miner (so the game now has $m-k+1$ miners), with each colluding miner receiving $1/k$ of the colluding utility, the utility achieved in an equilibrium with collusion is at most that achieved without collusion, assuming $m$ is sufficiently large. 
    \label{lemma:collusion}
\end{lemma}

In the heterogeneous cost model, it is unclear what collusion would mean for two miners with different costs, but one could imagine models where there are some miners with the same cost and they choose to collude. We leave this further analysis for future work. The general intuition we get from Lemma~\ref{lemma:collusion} is that with fewer miners, the equilibrium hashrate decreases thus the reward may increase as the cost decreases. So for the miners who don't collude, the equilibrium utility increases. But for miners who collude, they must then share the increased utility with all colluders, and it is unclear if the increase is enough to make up for splitting the utility into $k$ parts.

\smallskip
\noindent \textbf{Variable Coin Market Value.}
In Section~\ref{equilibrium_analysis}, we view the miner cost and reward in terms of the same currency unit. In reality, the reward is given in the coin of the blockchain being mined while cost is a real-world expense generally paid in the currency of the country where the mining is taking place. To bridge this gap we must understand how to convert real-world change in the price of the cryptocurrency to the relationship between the reward and the cost to miners. 

Consider the equilibrium analysis to be saying that a hashrate of $1$ for miner $i$ costs $c_i$ unit of cost (say dollars) and that one coin of the reward has $1$ unit of worth (i.e. $\$1$). Now, say the value of the currency changes by $R$, so one unit of currency is now worth $\$R$. We are now interested in understanding what happens to the equilibrium of the system, i.e. which miners would now participate at equilibrium and with what hashrate?

\begin{lemma}
    \label{lm:static-value}
    In the static-reward model, an increase in the value of the cryptocurrency by a factor of $R$ results in a new equilibrium strategy where the same miners participate with $Rq_i$ hashrate where $q_i$ is the previous equilibrium hashrate. The new system hashrate thus increases by a factor of $R$.
\end{lemma}

\begin{lemma}
\label{lm:happy_value}
    In \codename, an increase in the value of the cryptocurrency by a factor of $R$ results in the participation cost threshold to increase (allowing higher cost miners to participate), and the system hashrate to increase by a factor of $R$ until it reaches $Q$, then increase by a factor of $\sqrt[\delta+1]{R}$.
\end{lemma}

 \section{Discussion}
\label{discussion}
In this paper we've presented a novel family of mining reward functions which adjust to the hashrate of the system. Our functions fall in the class of \textit{generalized proportional allocation rules} of \cite{chen2019axiomatic} and thus inherit the properties of non-negativity, weak budget-balance, symmetry, sybil-proofness and collusion-proofness. These properties are defined based solely on the expectation of the reward of a miner and not under any equilibrium. In this work we've shown that for all $Q>0$ and $\delta \geq 0$ \codename has an equilibrium at a unique hashrate and set of miners, and if that hashrate is equal to $Q$ there may be multiple equilibria at $Q$. We further show that the equilibrium includes at least as many miners as the static-reward function and is at a hashrate at most that for the static-reward function. We also discuss collusion and sybil-proofness in equilibrium and that as the market value of the coin increases, the equilibrium shifts to include more miners at an increased hashrate that is sub-linear in the value of the coin after the system hashrate surpasses $Q$ (unlike the static-reward function whose equilibrium hashrate increases linearly indefinitely).   

\noindent\textbf{Long-term dynamics.} As our analysis focuses on equilibria, a natural question to ask is whether we introduce any unfavorable long-term dynamics by pegging our reward to the system hashrate. One such concern is on the control of supply of the system. Two current versions of coin issuance are the Bitcoin and Ethereum models. In Bitcoin the reward per block halves every 210K blocks (approximately every 4 years until it is 0), so that half the total supply ever was mined in the first 4 years. In Ethereum the block reward is set at 5 Ethers so that the total supply will never be capped. Our proposed model is novel in that assuming a steady increase in hashrate, the issuance will decrease smoothly over time. The rate of decrease, $\delta$, is a parameter set by the system designer.

In the start of any new cryptocurrency the coins have no value, thus the miners that initially mine are speculating that the coins will have value in the future making up for the cost. During this time the hashrate is generally low so the existing miners do not incur much cost.  When the currency does have more value, it appears older coins were mined for ``cheap". One could argue that those early miners mine speculatively, and for systems whose coin reward goes down over time, early miners may also control a large portion of the supply. The steeper the decline in the reward, the larger fraction of supply early miners control. As an example, it is estimated that the creator of Bitcoin, Satoshi Nakamoto, and assumed first miner, holds approximately 1 million Bitcoins \footnote{currently valued at 10 billion Dollars but which have never been spent and are assumed to stay out of circulation}, about 5\% of the total supply ever, probably mined at a cost of only a few dollars \cite{earlymining,satoshi}. 

As a currency grows in value, new miners are incentivized to start mining in the system until the cost to mine a block becomes close to the value of the reward for that block. Since the total supply of the currency is tied to the hashrate we get the interesting phenomena that as the system gains users (miners) the projected total supply decreases, but inversely, if the system decreases in value and starts to lose miners, \codename works a bit like a fail safe where the reward will increase and hopefully aid in incentivizing the remaining miners to stay, stabilizing the value of the system as opposed to a death spiral of miners leaving and the reward just losing value. In this paper, we model the utility of the miner as the per-block profit. To understand the long-term dynamics at play, a future analysis of the evolving game should incorporate market share into the utility of the miner and its impact on market centralization.

\smallskip
\noindent\textbf{Setting $Q$ and $\delta$} We show that an increase in $\delta$ comes with an increase in good decentralization properties we want, like more miners mining at equilibrium and big miners joining with less hashrate. The more you increase $\delta$ however, the more constrained the issuance of the currency becomes, which could lead to centralization in the market control to early adopters. Setting $Q$ and $\delta$ is thus a balancing game and involves practical considerations. 

The $\delta$ exponent in \codename controls how quickly the block reward declines. A low $\delta$ would correspond to a gradual decrease in the block reward as the hashrate increases. $Q$ is the threshold from which point the reward starts to decrease. One way to think of $Q$ is as a security lower-bound for the system.  When the hashrate reaches $Q$, any additional hashrate would lower the reward. A system designer should then choose a $Q$ based on the mining hardware of the system (e.g. ASICs,GPUs, etc.) and some understanding of likely advancements in its performance and choose $Q$ to be a conservative bound on the cost to amass enough hardware to attack the system (e.g. a 51\% attack). Based on this and the issuance rate the system designer is targeting a $\delta$ can be set.

Since any change to parameters in blockchain systems generally require a {\em hardfork} in the code, i.e. a change that breaks consensus between adopters and non-adopters, the Bitcoin model of blockchain software development is to avoid such changes unless absolutely critical. Other, more expressive systems (e.g. Ethereum and Zcash), have relied on hardforks to implement changes and increase functionality on a more regular basis. Though setting $Q$ and $\delta$ could be thoughtfully done only once in the inception of a new system, another approach would be to periodically update their values if the system's growth (both miner hashrate and value of the currency) is not within the predicted bounds. One such concern would be if the target hashrate $Q$ underestimated the growth of the system hashrate and thus stagnating the cost to attack the system. It would then be incentive compatible to increase $Q$ as it would incentivize higher hashrates (increase security) while also increasing the reward for the miners. One idea is to set $Q$ based on a long-term expected growth and have periodic updates (on the scale of years) to adjust $Q$ based on miner increase and mining hardware trends. 

 \section{Related work}
\label{related_work}

In this paper we've provided an equilibrium analysis of \codename, a new family of mining reward functions pegged to the network hashrate. As stated above, \codename is an example of the generalized proportional model of \cite{chen2019axiomatic}. We compare \codename with the equilibrium of the static reward function of \cite{arnosti2018bitcoin} associated with most cryptocurrencies. Other papers have looked at different games involved in mining including the game between participants in mining pools and different reward functions for how the pool rewards are allocated \cite{schrijvers2016incentive}. In \cite{li2019mean}, the authors present a continuous mean-field game for bitcoin mining which captures how miner wealth and strategies evolve over time. They are able to capture the ``rich get richer'' effect of initial wealth disparities leading to greater reward imbalances. 
\cite{iyidogan2019equilibrium} models the blockchain protocol as a game between users generating transactions with fees and miners collecting those fees and the block reward.
They show if there is no block reward, then there is an equilibria of transaction fee and miner hashrate. Higher fees incentivize higher miner hashrate which leads to smaller block times (in between difficulty adjustments). When you introduce a high static block reward, the users may no longer be incentivized to introduce mining fees and there may no longer be an equilibrium. 

In contrast, \cite{carlsten2016instability} also studies the case where there is no block reward, and analyzes new games in which miners may use transactions left in the mempool (pending transactions) to incentivize other miners to join their fork. Another work exploring the mining game when there is no block reward is that of \cite{tsabary2018gap} who introduce \textit{the gap game} to study how miners choose periods of times when not to mine (gaps) as they await more transactions (and their fees). They show that gap strategies are not homogeneous for same cost miners and that the game incentivizes miner coalitions reducing the decentralization of the system.

Previous work on rational attacks in cryptocurrency mining includes
\cite{badertscher2018but} who study the security of Bitcoin mining under rational adversaries using the Rational Protocol Design framework of \cite{garay2013rational} as a rational-cryptographic game.  Also, \cite{beccuti2017bitcoin} who analyze the Bitcoin mining game as a sequential game with imperfect information, and \cite{sapirshtein2016optimal} analyze selfish mining by looking at the minimal fraction of resources required for a profitable attack, tightening the previous lower-bounds and further extending the analysis to show how network delays further lower the computational threshold to attack. In \cite{kroll2013economics}, the authors explore the game of Bitcoin mining cost and reward focusing on incentives to participate honestly. 
They outline the choices different players can make in a blockchain system and their possible consequences, but their analysis does not take into account block withholding attacks. 
Another work related to the incentives at play in cryptocurrency mining is \cite{biais2019blockchain} which looks at the coordination game of Bitcoin miners in choosing which fork to build on when mining. 
They find the longest chain rule is a Markov Perfect equilibrium strategy in a synchronous network and  explore other miner strategies, some that result in persistent forks.

 \appendix
\section*{Appendix}
\renewcommand{\thesubsection}{\Alph{subsection}}

\subsection{Proof of Proposition~\ref{prop:diminishing_reward}}
\label{app:diminishing_reward}
\begin{proof}
First take the derivative of $U_i(q)$ w.r.t. $q_i$ and get
$$ \frac{\partial U_i(q)}{\partial q_i} = \frac{(q_i+1)^{\delta+1}-q_i(1+\delta)(q_i+1)^\delta}{(q_i+1)^{2(\delta+1)}}-c_i= \frac{1-q_i\delta}{(q_i+1)^{\delta+2}}-c_i$$
We set it equal to $0$ and get that for $0<c_i<1$ the only maxima is at $c_i=\frac{1-q_i\delta}{(q_i+1)^{\delta+2}}$ with $q_i\delta<1$. Taking the second derivative we get
$$ \frac{\partial^2 U_i(q)}{\partial q_i^2} =\frac{q_i\delta^2+q_i\delta-2\delta-2}{(q_i+1)^{\delta+3}}$$
so we get that the function is concave ($\frac{\partial^2 U_i(q)}{\partial q_i^2} < 0$) for $q_i\delta<2$, establishing (1).

For (2) we take the derivative of $ \frac{\partial U_i(q)}{\partial q_i}$ w.r.t. $\delta$ and get
$$ \frac{\partial^2 U_i(q)}{\partial q_i\partial \delta}=\frac{(q_i+1)^{\delta+2}(-q_i)-(1-q_i\delta)\ln(1+q_i)(q_i+1)^{\delta+2}}{(q_i+1)^{2(\delta+2)}}= \frac{-q_i+(q_i\delta-1)\ln(q_i+1)}{(q_i+1)^{\delta+2}}$$
which we want to show is negative so we get
$$-q_i+(q_i\delta-1)\ln(q_i+1)<0$$
which is negative at the maxima since $q_i\delta<1$. 
\qed
\end{proof}

\subsection{Hash-Pegged Reward Equilibrium Examples}
\label{something}

\subsubsection{General 2 miner example }\label{app:ex2}
In this example we set $\delta = 1$ and $Q = 1$ and have 2 miners with costs $c_1,c_2$ s.t. $c_1\leq c_2$. First we solve for $c^*$ s.t. $X(c^*) = \sum_i \max(1-c_i/c^*,0) = 1$. We get that $c^* = c_1+c_2$.

First, if $\frac{1}{c^*} = \frac{1}{c_1+c_2} < Q = 1$, i.e. $c_1+c_2 > 1$ then we can use the analysis of \cite{arnosti2018bitcoin} and get that $q_1 = \frac{1}{c_1+c_2}(1-c_1/(c_1+c_2))$ and $q_2 = \frac{1}{c_1+c_2}(1-c_2/(c_1+c_2))$ with $q_1+q_2= \frac{1}{c_1+c_2}$.

Next we analyze the case where $c_1+c_2\leq 1$. First assume that in equilibrium $q_1+q_2>Q=1$ with $q_1,q_2 >0$. We have that 
$$ U_i'(q) = \frac{q_1+q_2-2q_i}{(q_1+q_2)^3}-c_i = 0$$
so $c_1 = \frac{q_2-q_1}{(q_1+q_2)^3}$ and $c_2 = \frac{q_1-q_2}{(q_1+q_2)^3}$ therefore $(q_1+q_2)^3 = \frac{q_2-q_1}{c_1} = \frac{q_1-q_2}{c_2}$ and so $q_1 = q_2$. This could only happen if $c_1=c_2=c$. If so you have that the utility for both miners is $U_i(q) = 1/q-cq$ and $U_i'(q) = -1/q^2-c$ i.e. the utility is decreasing and would be maximized at the smallest $q$ which is $q<.5$ since we assumed $q_1+q_2>1$. Thus there is no equilibrium with this assumption.

Finally, we check the other case for $c_1+c_2< 1$ which is that any equilibrium strategy has $q_1+q_2=Q= 1$. Assume an equilibrium strategy where miner 1 puts in $\alpha$ and miner 2 puts in $\beta$ hashrate. We get that $U_1 = -c_1\alpha + \alpha/(\alpha+\beta)$ and the derivative w.r.t. $\alpha$ is $U_1'=-c_1+\beta/(\alpha+\beta)^2$. With $\alpha+\beta =1$ we get $U_i' = \beta-c_1$ which is only non-negative if $\alpha\leq 1-c_1$. Doing the same with $U_2$ w.r.t $\beta$ we get that $\beta \leq 1-c_2$. Assuming $\alpha+\beta \geq 1$ we get $U_1 = -c_1+\alpha/(\alpha+\beta)^2$ and the derivative w.r.t. $\alpha$ is $U_1' = -c_1+\frac{1}{(\alpha+\beta)^2}-\frac{2\alpha}{(\alpha+\beta)^3}$ and if we set $\alpha+\beta = 1$ we get $\alpha\geq \frac{1-c_1}{2}$. Doing the same with $U_2$ w.r.t $\beta$ we get $\beta\geq\frac{1-c_2}{2}$. Putting both bounds together we get that there are possibly many equilibrium strategies where $\alpha+\beta =1$ with $\frac{1-c_1}{2}\leq \alpha \leq 1-c_1$ and $\frac{1-c_2}{2}\leq \beta \leq 1-c_2$.

\subsubsection{Example with $c_i = \frac{i}{i+1}$}\label{app:ex4}
We now consider the case of cost function $c_i = \frac{i}{i+1}$ still considering $\delta = Q=1$. This case is interesting because we can solve for $c^* = \frac{493}{560} \approx .88$. This means that in the case of fixed reward, only the first 7 miners participate in equilibrium. Since this means $1/c^* > Q=1$, the old equilibrium point would now have less reward and thus may no longer be the equilibrium point.

To see this let's take the case where the lowest cost miner (i.e. miner $i=1$ with cost $0.5$) is deciding how much hashrate to buy knowing all other miners are using the old equilibrium strategy of $q_i = \frac{1}{c^*}\max(1-c_i/c^*,0) =  \frac{560}{493}~\max(1-c_i*560/493,0)$. So we have that the current sum of hashrate is 
$$H= \sum_{i=2}^7 1-\frac{i*560}{(i+1)493}=6-\frac{560}{493}(\frac{2}{3}+\frac{3}{4}+\frac{4}{5}+\frac{5}{6}+\frac{6}{7}+\frac{7}{8})\approx 0.6451374 < Q=1$$
This makes miner 1's utility function:
\begin{equation*}
U_1(q)= \begin{cases}
\frac{q_1}{H+q_1}-.5q_1 &\text{if $q_1 \leq 1-H$}\\
\frac{q_1}{(H+q_1)^2}-.5q_1 & \text{o/w}
\end{cases}
\end{equation*}

The old equilibrium strategy would have $q_1 = \frac{560}{493}(1-.5{560}{493}) \approx 0.490765$ giving utility $U_1(q)\approx 0.134974$. However, if we take $q_1 = 1-H$ we get a higher utility of $U_1(q)=0.1774312$, thus the old strategy doesn't work.

In the old strategy only the first 7 miners would participate in equilibrium, the next question is if this is still the case with the hash-pegged mining reward function. Let's assume the equilibrium point is at $\sum_j q_j = 1$ and not all miners participate at equilibrium. Take a miner $i$ who does not participate, this means $$U_i(q) = \frac{q_i}{(1+q_i)^2}-c_iq_i\leq 0$$ which implies
$$q_i(\frac{1}{(1+q_i)^2}-c_i)\leq 0~\text{ thus }~ c_i \geq \frac{1}{(1+q_i)^2}$$
which for $q_i=0$ means $c_i\geq 1$ but $c_i = \frac{i}{i+1}<1$ for all i. Thus if there is an equilibrium strategy with $\sum_j q_j =1$, all miners would participate.

We follow this and assume there is an equilibrium strategy with $\sum_j q_j=1$ for $n$ miners. Taking the derivative of the utility for a miner $i$ we get that at equilibrium we have 
\begin{equation*}
U_i'(q)= \begin{cases}
1-q_i-c_i\geq 0 & \text{ for } \sum_j q_j \leq 1\\
1-2q_i-c_i\leq 0 & \text{ for } \sum_j q_j \geq 1
\end{cases}
\end{equation*}
Using the bottom inequality we get for and equilibrium strategy, for each $i$ we have $q_i\geq \frac{1}{2}(1-c_i)$. Summing up the lower bound for each $q_i$ gives us that $1=\sum_j^n q_j \geq \frac{1}{2}\sum_j^n\frac{1}{i+1}$ which can only hold for $n\leq 10$. For $n>10$, an equilibrium strategy must thus have that $\sum_j q_j >1 $, we analyze this next. 

Since $\sum_j q_j > 1$ we can take the utility function $U_i(q) = \frac{q_i}{(\sum_j q_j)^2}-c_iq_i$ with derivative $U_i'(q) = \frac{1}{(\sum_j q_j)^3}(\sum_j q_j-2q_i)-c_i$. Assuming only the first $n$ miners participate at equilibrium, set $H = \sum_j q_j >1$ and solve the $q_i$ for $i\leq n$, we get
$$\frac{1}{H^3}(H-2q_i)-\frac{i}{i+1} = 0$$
$$q_i = \frac{H}{2}(1-H^2c_i)$$
Summing over all $q_i$ for $i\leq n$ we get
$$H = \sqrt{\frac{n-2}{\sum_{j=1}^n \frac{j}{j+1}}}$$
So the equilibrium strategy has 
$$ q_i = \frac{1}{2}\sqrt{\frac{n-2}{\sum_{j=1}^n \frac{j}{j+1}}}(1-\frac{(n-2)i}{\sum_{j=1}^n \frac{j}{j+1}(i+1)})  $$
for all miners that participate in equilibrium. We can iterate over $n$ to find that with this strategy, equilibrium exists with $n=25$.

\subsubsection{All miners have cost $c$}\label{app:ex3}
The next example we consider is the case of homogeneous cost with $m$ miners with cost $c$. For this case we will use $Q=1$ and \codename with any $\delta$. We can solve for $c^*$ s.t. $X(c^*) = \sum_i\max(1-c/c^*,0)= m(1-c/c^*)=1$, we get $c^* = \frac{mc}{m-1}$. For $\frac{1}{c^*} < Q=1$ (i.e. $c> \frac{m-1}{m}$) we can use the analysis of \cite{arnosti2018bitcoin} and get that $q_i =\frac{1}{c^*}(1-c/c^*) = \frac{m-1}{m^2c}$ with $\sum_i q_i = \frac{m-1}{mc} < Q=1$.

We next consider the other case, i.e. $c\leq\frac{m-1}{mQ}= \frac{m-1}{m}<1$. We first consider the case of $\sum_i q_i >Q=1$. We get the following utility function $U_i(q) = \frac{q_i}{(\sum_i q_i)^{\delta +1}}-cq_i $ with derivative
$$U_i^{'}(q) = \frac{\sum_j q_j - (\delta+1)q_i}{(\sum_jq_j)^{\delta+2}}-c$$ 
setting this equal to 0 with a homogeneous strategy we get $q_i = \frac{1}{m}\sqrt[\delta+1]{\frac{m-\delta-1}{cm}}$. With this we get that $\sum_i q_i = \sqrt[\delta+1]{\frac{m-\delta-1}{cm}}$. For this to be $>1$ we get that $c<\frac{m-\delta-1}{m}$.

Finally we consider the final case where $\frac{m-\delta-1}{m} \leq c \leq \frac{m-1}{m}$. The only strategy left if for $H = \sum_j q_j =Q=1$. Lets consider the case where all but one miner invest $\frac{1}{m}$, the utility function for that last miner is thus
\begin{equation*}
U_i(q)= \begin{cases}
\frac{q_i}{\frac{m-1}{m}+q_i}-cq_i &\text{if $q_i \leq 1/m$}\\
\frac{q_i}{(\frac{m-1}{m}+q_i)^{\delta+1}}-cq_i & \text{if $q_i \geq 1/m$}
\end{cases}
\end{equation*}

we take the derivative w.r.t. $q_i$ 

\begin{equation*}
U_i'(q)= \begin{cases}
\frac{\frac{m-1}{m}}{(\frac{m-1}{m}+q_i)^2}-c &\text{if $q_i \leq 1/m$}\\
\frac{\frac{m-1}{m}-\delta q_i}{(\frac{m-1}{m}+q_i)^{\delta+1}}-c & \text{if $q_i \geq 1/m$}
\end{cases}
\end{equation*}

For $q_i\leq 1/m$ we have that $U_i'(q) = \frac{\frac{m-1}{m}}{(\frac{m-1}{m}+q_i)^2}-c \geq \frac{m-1}{m}-c \geq 0$ since $c\leq\frac{m-1}{m}$. 
For  $q_i\geq 1/m$ we have that $U_i'(q)= \frac{\frac{m-1}{m}-\delta q_i}{(\frac{m-1}{m}+q_i)^{\delta+1}}-c \leq \frac{m-\delta -1}{m}-c \leq 0$  
since $c\geq \frac{m-\delta -1}{m}$.
Thus $U_i'$ is positive for $q_i \leq 1/m$ and negative for $q_i \geq 1/m$ therefore the equilibrium strategy of miner $i$ is $q_i = 1/n$.

With the static reward function, the equilibrium hashrate of the system  is $\sum_i q_i = 1/c^* = \frac{m-1}{mc}$. With \codename, we have 3 cases: (i) if $c> \frac{m-1}{m}$, then the equilibrium hash-rate is the same as the static reward function, (ii) if $\frac{m-\delta-1}{m}\leq c \leq \frac{m-1}{m}$ then the equilibrium hashrate is $Q=1\leq \frac{m-1}{mc}=$ the static reward hashrate. Finally, (iii) if $c< \frac{m-\delta-1}{m}$, then the equilibrium system hashrate is $\sqrt[\delta+1]{\frac{m-\delta-1}{cm}}<1< \frac{m-1}{cm}$. Thus, given a $c$, \codename with any $\delta$ has equilibria with system hashrate less than or equal to the static reward equilibria.  

\subsection{Proofs for Section~\ref{equilibrium_impacts}}
\label{app:equilibrium_impacts}
\noindent \textbf{Resistance to collusion attacks.}

\begin{LabeledProof}{Lemma~\ref{lemma:collusion}}
From Theorem~\ref{thm:equillib} we take $c^*$ s.t. $X(c^*) = 1$ and $c^\dagger$ s.t. $X(c^\dagger)=2$ and get
$$ c^* = \frac{cm}{m-1} ~~~\text{and}~~~ c^\dagger = \frac{cm}{m-\delta-1}.$$
We now have 3 non-collusion equilibrium cases to compare against: $c^*>1/Q$, $c^*\leq 1/Q\leq c^\dagger$ and $c^\dagger >1/Q$:

\paragraph{Case 1 ($c^*>1/Q$)} We have that $\frac{cm}{m-1}>1/Q$ meaning that at equilibrium $\sum_i q_i= \frac{1}{c^*}<Q$ with $U_i(q) =  \frac{1}{m^2}$ for all miners. Now consider that $k$ miners collude so that there are now $k-n+1$ miners. We have that $c^*_{new} = \frac{c(m-k+1)}{m-k}> \frac{1}{Q}$ since $c^*_{new}$ is increasing in $k$. Thus the new equilibrium utility for the colluders is
$$U_{i,k}(q)= \frac{1}{k(m-k+1)^2}.$$ 
We check if $U_{i,k}(q) \geq U_i(q)$ and get 
$$k\geq m-\frac{1}{2}\sqrt{4m+1}+1/2 > m/2$$
meaning a majority of miners must collude for there to be a non-negative utility gain. 

\paragraph{Case 2 ($c^*\leq1/Q\leq c^\dagger$)} 
We first note that $c^* = \frac{cm}{m-1}\leq 1/Q$ means that $c<1/Q$ and $m\geq \frac{1}{1-cQ}$. Next we note that $c^*_{new} = \frac{c(m-k+1)}{m-k}$ and $c^\dagger_{new} =\frac{c(m-k+1)}{m-k-\delta}$ are increasing as $k$ increases. Thus it will never be the case that $c^\dagger_{new}<1/Q$. We must thus consider just two cases (i) $c^*_{new}\leq 1/Q \leq c^\dagger_{new}$ and (ii) $c^*_{new} > 1/Q$. We compare the collusion utility with each miner's utility before collusion which is $U_i(q) = \frac{1}{m}(1-cQ)$. \\

\noindent(i) $c^*_{new}\leq 1/Q \leq c^\dagger_{new}$. Since the miners are still in the equilibrium regime such that the sum of hash rates will be Q, we get the new equilibrium for each miner that colludes to be $U_{i,k} = (\frac{1}{k})\frac{1}{m-k+1}(1-cQ)$. Comparing this to their non-colluding equilibrium, we get that collusion is only beneficial if
$$(\frac{1}{k})\frac{1}{m-k+1}>\frac{1}{m}$$
$$(k-m)(k-1)>0$$
which is never true for $m>k>1$, thus collusion is not beneficial.\\

\noindent(ii) $c^*_{new} > 1/Q$. The new utility for the miners that collude is $U_{i,k}(q)=(\frac{1}{k})\frac{1}{(m-k+1)^2}$. Note the non-colluding utility is $U_i(q) = \frac{1}{m}(1-cQ) \geq \frac{1}{m^2}$ since $m\geq\frac{1}{1-cQ}$. We are interested if $U_{i,k}(q) \geq U_i(q)$, i.e.
$$ \frac{1}{k(m-k+1)^2}  \geq \frac{1}{m}(1-cQ)\geq \frac{1}{m^2}$$
which we saw from \textit{Case 1} is not satisfiable for $k<m/2$.

\paragraph{Case 3 ($c^\dagger<1/Q$)} 
We start with the utility for each miner without collusion to be $U_i(q) = \frac{\sqrt[\delta+1]{Q^\delta}}{(\delta+1)\sqrt[\delta+1]{c^\dagger}}(1-c/c^\dagger)(c^\dagger-c)$. As with case 2, since $c^*<1/Q$ we know $c<1/Q$, and for $c^\dagger$ to be defined it must be that $m>\delta+1$. We now must handle each case for the collusion equilibrium (i) $c^\dagger_{new}<1/Q$ (ii) $c^*_{new}\leq 1/Q\leq c^\dagger_{new}$ and (iii) $c^*_{new}>1/Q$.
[note: we also have that $\delta<\frac{m-1-cQ}{1+cQ}$]
\\
(i) $c^\dagger_{new}<1/Q$. The collusion equilibrium is thus
$$ U_{i,k}(q) = \frac{1}{k}\frac{\sqrt[\delta+1]{Q^\delta}}{(\delta+1)\sqrt[\delta+1]{c^\dagger_{new}}}(1-c/c^\dagger_{new})(c^\dagger_{new}-c)$$
letting $c^\dagger_{new} = \frac{c(m-k+1)}{m-k-\delta}$ we can re-write the utility as
$$ U_{i,k}(q) = \frac{(\delta+1)}{k}\sqrt[\delta+1]{\frac{Q^\delta c^\delta}{(m-k+1)^{\delta+2}(m-k-\delta)^\delta}} .$$
Note that for small $k$ the denominator is increasing ($U_{i,k}(q)$ is decreasing) and that for $k<m$ it will either keep increasing or flip once to decreasing. We thus check if $U_{i,k}(q)$ for $k=m/2$ is larger than for $k=1$(non-collusion):
$$ \frac{(\delta+1)}{m/2}\sqrt[\delta+1]{\frac{Q^\delta c^\delta}{(m/2+1)^{\delta+2}(m/2-\delta)^\delta}}\geq (\delta+1)\sqrt[\delta+1]{\frac{Q^\delta c^\delta}{(m)^{\delta+2}(m-1-\delta)^\delta}} $$
which simplifies to
$$2^{3(\delta+1)}m(m-1-\delta)^\delta \geq (m+2)^{\delta+2}(m-2\delta)^\delta.$$
The above inequality fails to hold if $m$ is large enough, i.e., if $m \ge 2\delta + 8$.
\junk{
other inequality to consider
$$m^{\delta+2}(m-1-\delta)^\delta\geq k^{\delta +1}(m-k+1)^{\delta+2}(m-k-\delta)^\delta$$
**need to finish this proof**}
    \junk{
        (i)$c^\dagger_{new}<1/Q$ gives us that $k<m-\frac{1+cQ}{1-cQ}$. Since $c^\dagger>c^*$ we have that $c^*<1/Q$ which gives us $\frac{cm}{m-1}<1/Q$ which we can re-write to be $m>\frac{1}{1-cQ}$. 
        
        In general, $U_{i,k} = \frac{\sqrt{Q}}{2k\sqrt{c^\dagger_{new}}}(1-c/c^\dagger_{new})(c^\dagger_{new}-c) = \frac{1}{2kc^\dagger_{new} \sqrt{c^\dagger_{new}}}(c^\dagger_{new}-c)^2$, where 
        $c^\dagger_{new} = c(1 + 2/(n-k-1))$.  Substituting for $c^\dagger_{new}$, we get
        
        \[
        U_{i,k} = \sqrt{Q}\frac{4c^2 (m-k-1)^{3/2}}{2k c^{3/2} (m-k+1)^{3/2} (m-k-1)^2} = \sqrt{Q}\frac{2\sqrt{c}}{k(m-k+1)^{3/2} \sqrt{m-k-1}}.\]
        
        Note that for small $k$ the denominator is increasing (i.e. $U_i$ is decreasing) and that for $k<m$ it will either keep increasing or flip once to decreasing. We thus compare the value of $U_i$ for $k=1$ and $k=m/2$ and see if the later is larger as follows
        $$ \sqrt{Q}\frac{2\sqrt{c}}{m^{3/2}\sqrt{m-2}} < \sqrt{Q}\frac{2\sqrt{c}}{m/2(m/2+1)^{3/2}\sqrt{m/2-1}}$$
        $$m(m/2+1)^{3/2}\sqrt{m/2-1}<2m^{3/2}\sqrt{m-2}$$
        $$m^2(m/2+1)^3(m/2-1)<4m^3(m-2)$$
        $$(m+2)^3<64n$$
        which does not hold for $m>5$, thus colluding is not beneficial with a minority colluders for $m>5$.
    }
\\
(ii) $c^*_{new}\leq 1/Q\leq c^\dagger_{new}$. The new utility with collusion is $$U_{i,k}(q)=\frac{1}{k(m-k+1)}(1-cQ).$$
\junk{
Taking the derivative of the denominator w.r.t. $k$ we get
$$\frac{\partial}{\partial k}k(m-k+1) = m-2k+1$$
which is positive for $k<\frac{m+1}{2}$, meaning the utility is decreasing for smaller values of $k$ and increasing for $k>(m+1)/2$. We thus check if the new utility at $k=2$ is larger than the non-colluding utility:
$$ \frac{(1-cQ)}{2(m-1)}\geq (\delta+1)\sqrt[\delta+1]{\frac{Q^\delta c^\delta}{(m)^{\delta+2}(m-1-\delta)^\delta}}$$
}
From $c^\dagger_{new} \geq 1/Q$ for any $k$ we get that  $\frac{1}{1-cQ}\geq \frac{m-k+1}{\delta+1}$ and $cQ \geq \frac{m-k-\delta}{m-k+1}$. 
We substitute those and obtain the following condition for collusion.
$$m^{\delta+2}(m-1-\delta)^{\delta}\geq k^{\delta+1}(m-k+1)^{\delta+2}(m-k-\delta)^\delta$$
For $m$ sufficiently large, the RHS of the above inequality is minimized at $k = 2$; this implies that the above inequality cannot hold as long $m \ge \delta + 9$.
\junk{
** finish proof by either the above being false, or only true for $k>m/2$**}
    \junk{
        (ii)$c^*_{new}\leq 1/Q\leq c^\dagger_{new}$ has that the new utility with collusion is $U_{i,k}=\frac{1}{k(m-k+1)}(1-cQ)$. Taking the derivative of the denominator w.r.t. $k$ we get
        $$\frac{\partial}{\partial k}k(m-k+1) = m-2k+1$$
        which is positive for $k<\frac{m+1}{2}$, meaning the utility is decreasing for smaller values of $k$ and increasing for $k>(m+1)/2$. We thus check if the new utility at $k=2$ is larger than the non-colluding utility:
        $$U_{i,k}=\frac{1-cQ}{2(m-1)}>\frac{2\sqrt{cQ}}{m\sqrt{m}\sqrt{m-2}}=U_i$$
        $$\frac{m^{3/2}\sqrt{m-2}}{(m-1)}>\frac{4\sqrt{cQ}}{1-cQ}$$
        Note that for $k=2$, $c^\dagger_{new} = \frac{c(m-1)}{m-3}\geq 1/Q$ gives us that $\frac{1}{1-cQ}\geq \frac{m-1}{2}$ and $cQ\geq \frac{m-3}{m-1}$. Substituting those on the right-hand side, we get
        $$\frac{m^{3/2}\sqrt{m-2}}{(m-1)}>\frac{2\sqrt{m-3}(m-1)}{\sqrt{m-1}}$$
        $$\frac{m^{3/2}\sqrt{m-2}}{(m-1)^{3/2}\sqrt{m-3}}>2$$
        which is only satisfiable for around $4.26>m>3$, thus colluding is not more beneficial for a minority colluders. 
    }
\\
(iii) $c^*_{new}>1/Q$. We have that the new utility is $U_{i,k} = \frac{1}{k(m-k+1)^2}$, implying the following condition for collusion.
\junk{
which is decreasing for $k<\frac{m+1}{3}$ and $k>m+1$. }
$$\frac{1}{k(m-k+1)^2}\geq (\delta+1)\sqrt[\delta+1]{\frac{Q^\delta c^\delta}{(m)^{\delta+2}(m-1-\delta)^\delta}}$$
$$ m^{\delta+2}(m-1-\delta)^\delta\geq k^{\delta+1}(m-k+1)^{\delta+2}(m-k)^\delta$$
This condition is weaker than the one derived for (ii); so we again obtain that collusion cannot occur as long as $m \ge \delta + 9$.

\junk{**show it can't be satisfied**\\
Next we check at $k=m/2$ and show again a contradiction for reasonable values.}
    \junk{
        (iii) $c^*_{new}>1/Q$ we have that the new utility is $U_{new} = \frac{1}{k(m-k+1)^2}$ which is decreasing for $k<\frac{m+1}{3}$ and $k>m+1$. So we first check if the value at $k=2$, $U_{i,k} = \frac{1}{(m-1)^2}$, is greater than the non-colluding equilibrium (note that $k$ cannot equal $1$ and $c^*_{new}>1/Q$ with $c^\dagger>1/Q$)
        $$ U_{i,k}=\frac{1}{2(m-1)^2}>\frac{2\sqrt{cQ}}{m\sqrt{m}\sqrt{m-2}}=U_i$$
        $$\frac{m^{3/2}(m-2)^{3/2}}{(m-1)^2}>4\sqrt{cQ}>4\frac{\sqrt{m-2}}{\sqrt{m-1}}$$
        since from $c^*_{new} = \frac{c(m-1)}{m-2}>1/Q$ we get  $cQ>\frac{m-2}{m-1}$. Thus we get
        $$\frac{m^{3/2}}{(m-1)^{3/2}}>4$$
        
        Thus the above can only be satisfied with around $1.2>m$ and therefore $U_{i,k}$ at $k=2$ is less than $U_i$. We now check for $k=m/2$:
        $$U_{i,k} = \frac{1}{m/2(m-m/2+1)^2}=\frac{8}{m(m+2)^2}>\frac{2\sqrt{cQ}}{m\sqrt{m}\sqrt{m-2}}=U_i$$
        $$ \frac{16}{(m+2)^4}>\frac{c}{m(m-2)}$$
        Note that $c^*_{new}>1/Q$ for $k=m/2$ gives us that $cQ>\frac{m}{m+2}$ so we get
        $$\frac{16}{m(m+2)^2}>\frac{16}{(m+2)^3}>\frac{cQ(m+2)}{m(m-2)}>\frac{1}{m-2}>\frac{1}{m}$$
        and we get $2>m$. Thus colluding is not beneficial with a minority of colluders.
    }

\end{LabeledProof}

\noindent \textbf{Resistance to Sybil attacks.}
First, we consider a Sybil attack in the static reward case where one miner pretends to be $k$ miners. Instead of utility $U_i =\frac{1}{m^2}$, this miner would get utility $U_s = k\frac{1}{(m+k-1)^2}$. We want to solve for the case when
$$ U_i < U_s$$
$$\frac{1}{m^2} < k\frac{1}{(m+k-1)^2}$$
$$ k< m^2-2m+1$$
it is thus more profitable to pretend to be multiple miners and arrive at an equilibrium with higher total hashrate and utility. Taking $U_i = \frac{k}{(m+k-1)^2}$ we get $U_i^{'} = \frac{m-k-1}{(m+k-1)^3}$ which is maximized at $k= m-1$, i.e. the miner gets most utility being $m$ miners.

If it is optimal for a single miner, then each miner may want to optimize their utility this way. We now consider a new game where we begin with $m$ players each with cost $c$ and each player $i$ decides how many miners $k_i$ they want to be and the hashrate $q_{i,j}$ for $j\in[1,k]$ each miner will have. Note that each miner needs to act independently to arrive at a globally optimal equilibrium as described above. 

If there are $N$ other miners, player $i$ will choose $k$ miners s.t. 
$$U_i = \frac{k}{(N+k)^2}$$
is maximized. Taking the derivative of the utility we get 
$$U_i^{'} = \frac{N-k}{(N+k)^3}$$
which is maximized at $k=N$.

Say all players choose the same number of miners then they would be trying to maximize $U_i = \frac{k}{(m*k)^2}=\frac{1}{m^2k}$ which maximized at the minimum $k$ values which is 1. But each player individually would choose $k=N$ for each $N$ meaning that the game would continue indefinitely with each player continuous choosing to match the total number of other players which locally increases their utility but brings the game to each miner's utility approaching 0.
 
 This is an example of the tragedy of the commons, where each player locally optimizing brings the utility of the whole system (and themselves) to it's minimum. 
 
The problem with the above game is that we are considering the case where a player chooses a $k$ based on their equilibrium utility with that $k$. In actuality, if a player has a single miner and the game is in equilibrium with their hashrate being $q$ and all other players hashrate summing to $H$, if this player at that moment pretended to instead be $k$ miners with some $q_j$ hashrate for each of it's miners $j$ s.t. $\sum_j q_j =q$ their utility would be
$$ U = \sum_j q_j(\frac{1}{q+H}-c) = q(\frac{1}{q+H}-c)$$
i.e. their utility in the moment does not change. This same logic applies to the hash-pegged utility case. In fact, if there is any overhead to Sybil, then there is no reason to do a Sybil. 

\smallskip
\noindent \textbf{Variable Coin Market Value.}
\begin{LabeledProof}{Lemma~\ref{lm:static-value}}
    Consider first the static reward function with the block reward being $\$1$. 
        In this model, the equilibria has system hashrate $H=1/c^*$, and all miners with $c_i <c^*$ participate with hashrate $q_i = \frac{1}{c^*}(1-c_i/c^*)$ and utility $U_i(q)=(1-c_i/c^*)^2$. 
    
    Say the block reward is now worth $\$R$, the new utility is $U^{new}_i(q)= \frac{Rq_i}{\sum_j q_j}-q_ic_i = R[\frac{q_i}{\sum_j q_j}-q_ic_i/R]$. Let $c^{new}_i = c_i/R$, the utility is $U^{new}_i(q) = R[\frac{q_i}{\sum_j q_j}-q_ic^{new}_i]$, since $R$ is just an outside multiple, this utility has the same equilibria as ${\sum_j q_j}-q_ic^{new}_i$. Thus we can solve for $c^*_{new}$ s.t. $\sum_j \max(1-c^{new}_i/c^*_{new})=1$ and get that $c^*_{new} = c^*/R$. We thus get that the new system hashrate is $H = \frac{1}{c^*_{new}}=R/c^*$ and all of the same miners participate with $q^{new}_i = R*q_i$ and utility $U^{new}_i(q) = R*U_i(q)$. Thus in the static reward case, all of the same miners participate with $R$ times the hashrate and utility.
\end{LabeledProof}

\begin{LabeledProof}{Lemma~\ref{lm:happy_value}}
    With \codename , we get that the new utility is 
    \begin{equation*}
    	U^{new}_i(q)= \begin{cases}
    		R*[\frac{q_i}{\sum_j q_j}-q_ic^{new}_i] &\text{if $\sum_j q_j \leq Q$}\\
    		R*[\frac{q_i}{(\sum_j q_j)}(\frac{Q}{\sum_j q_j})^\delta-q_ic^{new}_i] &~~~~~~~\text{o/w}
    	\end{cases}
    \end{equation*}
    Like in Lemma~\ref{lm:static-value}, we get that $c^*_{new} = c^*/R$ and similarly $c^\dagger_{new}=c^\dagger/R$. The difference is how these two new values compare with $1/Q$ which decides which case of the equilibria we end up in. For $R>1$, we get the following cases:\\
    (i) If $c^\dagger<1/Q$, then we remain in the case where the equilibria hashrate is $>Q$ but it goes from $H=\sqrt[\delta+1]{\frac{Q^\delta}{c^\dagger}}$ to $H_{new} = \sqrt[\delta+1]{\frac{Q^\delta R}{c^\dagger}}$, with the same miners participating.
    \\
    (ii) If $c^*\leq 1/Q \leq c^\dagger$ , then we have two possibilities:
    If $c^*_{new}\leq 1/Q \leq c^\dagger_{new}$, then the equilibria still has hashrate $Q$ but now with all miners with $c_i/R < 1/Q$ participating which could include more miners. The second case is that now $c^\dagger_{new}<1/Q$, we get that the new system hashrate is $H = \sqrt[\delta +1]{\frac{Q^\delta R}{c^*}}>Q$ and all miners with hashrate $c_i<c^\dagger$ participate which is strictly greater than or equal to the number of miners participating before. 
    \\
    (iii) If $c^*>1/Q$, we get the last case which includes 3 possibilities. First, if $c^*_{new}>1/Q$ then we remain with the system hashrate less than $Q$ (though $R$ times what it was before) and we get the case of the static reward function where the same miners participate but now with $R$ the hashrate and utility. The second case is that now $c^*_{new}\leq 1/Q \leq c^\dagger_{new}$  meaning the equilibrium system hashrate is $Q$ and all miners with $c_i/R<1/Q$ participate which is at least as many as before since $c^*_{new}\leq 1/Q$ means $c^*<R/Q$. The last case is that now $c^\dagger_{new}>1/Q$ so the system hashrate is now over $Q$ and all miners with $c_i<c^\dagger$ participate which is again at least as many miners as before.
\end{LabeledProof}

Thus what we get is that as the value of the currency (and therefore the value of the block reward) increases, the \codename equilibrium shifts so that the cut-off cost for miner participation increases. We also get an equilibrium system hashrate increase, which is linear in $R$ until the system hashrate reaches the bound $Q$ then it becomes linear in $\sqrt[\delta+1]{R}$. 

\end{document}